 \documentclass[onecolumn]{IEEEtran}
\onecolumn
\usepackage{amsfonts}
\usepackage{mathrsfs}
\usepackage{amsmath}
\usepackage{amssymb}
\usepackage{graphicx}
\usepackage{epic}
\usepackage{epstopdf}
\usepackage{color}
\usepackage{cite}
\usepackage{mathtools}
\usepackage{enumerate}
\usepackage{lineno}
\usepackage{url}
\usepackage{float}
\usepackage{color}
\usepackage{bm}
\usepackage{tikz}
\tikzstyle{vecArrow} = [thick, decoration={markings,mark=at position
 1 with {\arrow[semithick]{open triangle 60}}},
double distance=1.4pt, shorten >= 5.5pt,
preaction = {decorate},
postaction = {draw,line width=1.4pt, white,shorten >= 4.5pt}]
\tikzstyle{innerWhite} = [semithick, white,line width=1.4pt, shorten >= 4.5pt]
\renewcommand{\paragraph}{\roman{paragraph}}

\newtheorem{theorem}{Theorem}[section]
\newtheorem{corollary}[theorem]{Corollary}

\newtheorem{question}[theorem]{Question}
\newtheorem{problem}[theorem]{Problem}
\newtheorem{lemma}[theorem]{Lemma}

\newtheorem{Remark}[theorem]{Remark}
\newenvironment{proof}{\noindent {\bf Proof.}}{\rule{3mm}{3mm}\par\medskip}

\newcommand{\C}{\mathcal {C}}

\begin{document}

\title{Optimal Fraction Repetition Codes for Access-Balancing in Distributed Storage}

\author{Wenjun Yu, Xiande Zhang and Gennian Ge

\thanks{W. Yu ({\tt yuwenjun@mail.ustc.edu.cn}) and X. Zhang ({\tt drzhangx@ustc.edu.cn}) are with School of Mathematical Sciences,
University of Science and Technology of China, Hefei, 230026, Anhui, China. The research of X.Zhang was supported by the National Natural Science Foundation of China under Grant No. 11771419.  }
\thanks{G. Ge ({\tt gnge@zju.edu.cn}) is with  the School of Mathematical Sciences, Capital Normal University,
Beijing 100048, China. The research of G. Ge was supported by the National Natural Science Foundation of China under Grant No. 11971325, National Key Research and
Development Program of China under Grant Nos. 2020YFA0712100  and  2018YFA0704703, and Beijing Scholars Program.
}}
\date{}
\maketitle

\begin{abstract}
To solve the access-balancing problem in distributed storage systems, we introduce a new combinatorial model, called MinVar model for   fractional repetition (FR) codes. Since FR codes are based on graphs or set systems, our MinVar model is characterized by the property that the variance among the sums of block-labels incident to a fixed vertex is minimized. This characterization is different from Dau and Milenkovic's MaxMinSum model, while the minimum sum of labels is maximized. We show that our MinVar model is meaningful by distinguishing labelings with different variances  but with the same MaxMin value for some FR codes. By reformulating the MinVar model to an equivalent vertex-labeling problem of graphs, we find several families of optimal FR codes with balanced access frequency, and provide fundamental results for both problems. It is interesting that MinVar model is closely related to the concept of magic-labeling in graph theory.

\end{abstract}

{\bf Keywords:} Distributed  storage; access-balance; fractional repetition codes; access-variance; magic-labeling

\section{Introduction}
Motivated by the access-balancing issue in the coding for distributed  storage systems \cite{dimakis2010network}, Dau
and Milenkovic \cite{maxminsum} introduced a problem of  labeling
the points of the underlying combinatorial designs. In this framework, a file is split into several equal-sized parts and encoded into data chunks by an outer MDS code. After this, each data chunk is replicated a certain number (\emph{replication number}) of times and distributed among multiple storage nodes based on an inner  \emph{fractional repetition (FR) code} \cite{rashmi2009explicit,shah2011distributed,el2010fractional}. The combination of the outer MDS code and the inner FR code supports redundancy and reparability of the storage system, and constitutes a class of \emph{minimum bandwidth regenerating} (MBR) codes \cite{dimakis2010network} with the property of exact \emph{repair by transfer}. The problem of balancing data placements and  loads  of the storage in such a scheme, requires a constant data replication number and  a constant node volume of the inner FR code. This is the main reason that  combinatorial designs such as Steiner systems are commonly employed in data placement \cite{OlmezFractional,olmez2012repairable,zhu2014general}. Systems like Hadoop Distributed File System and  Google File System apply this strategy \cite{cidon2013copysets}.

 Access balancing aims to balance the access
requests to the nodes by using  data chunk popularity information \cite{cherkasova2004analysis}. In Dau and Milenkovic's model \cite{maxminsum},  the data chunks are labeled by popularity, and the overall popularities of chunks stored on each node  need to be balanced. That is, to find a proper labeling of the underling combinatorial design such that the sums of  labels in each block are as equal as possible. In particular, they defined  functions of designs to measure this property,  MaxMin (or MinMax), the  maximum (minimum) value of the minimum (maximum) block-sum in the design, and successfully found all Steiner triple systems that achieve the MaxMin value. This problem was further studied in \cite{brummond2019kirkman} for Kirkman systems and in \cite{chee2019access} for partial Steiner systems.

Although combinatorial designs are commonly used as underlying structures of the storage scheme, they are not usually the best choice for FR codes in general.  Regarding to the maximum size of the file that can be stored in a DRESS code \cite{pawar2011dress,el2010fractional}, Silberstein and  Etzion \cite{silberstein2015optimal} studied optimal FR codes based on graphs and designs. They constructed two kinds of optimal FR codes with replication number two, one is based on  Tur\'an graphs and the other is based on graphs with large girth. For bigger replication number, they showed that transversal designs and generalized polygons can produce optimal FR codes.

\par In this paper, we focus on the access-balancing problem for optimal FR codes. By observing that the MaxMinSum model only cares about the minimum sum of labels, we introduce a new model which considers the variance of the sums of labels. The new model is called MinVar model, which aims to approach the minimum access-variance of overall popularities among all nodes. This is in fact a block-labeling problem of set systems, such that the variance of the sums of labels of blocks incident to any fixed vertex is minimized. When the minimum variance attains zero, the problem is indeed a magic labeling problem for graphs, from which we can find several optimal FR codes with balanced access requests. Our second contribution is introducing an equivalent problem, which is a vertex-labeling problem of graphs when the set system is linear. By solving this problem for special graphs, we estimate the minimum access-variance of the MinVar model for several optimal FR codes.

The paper is organized as follows. Section~\ref{s:2} reviews FR codes, set systems, graphs and their relations. Section~\ref{s:3} introduces the  MinVar model, its relation to the magic labeling problem, and its equivalent vertex-labeling problem of graphs. In Section~\ref{s:4}, we solve the  equivalent vertex-labeling problem for several graphs, which helps to attack the MinVar problem in Section~\ref{s:5}. Concluding remarks are provided in Section~\ref{s:6}.
\section{Preliminaries and Notations}\label{s:2}
In this section we provide useful definitions of codes, graphs and set systems, and their relations among each other.

\subsection{Fractional repetition codes}

El Rouayheb and Ramchandran \cite{el2010fractional} introduced the concept of  DRESS (Distributed
Replication based Exact Simple Storage) code, which consists
of the concatenation of an outer MDS code and an inner
FR code.

Let $[\theta]:=\{1,2,\ldots,\theta\}$. Assume that $n,\alpha,\theta,\rho$ are positive integers satisfying $n\alpha=\theta\rho$. An $(n,\alpha,\rho)$ FR code $\C$ is a collection of $n$ subsets of $[\theta]$, $N_1,N_2,\ldots,N_n$, each of size $\alpha$ such that, each symbol of $[\theta]$ appears in exactly $\rho$ subsets of $\C$.
A $[(\theta, M), k,(n, \alpha, \rho)]$ DRESS code is a code consisting of an outer $(\theta, M)$ MDS code and an inner $(n,\alpha,\rho)$ FR code $\C$. First, a file $\mathbf{f}=(x_1,x_2,\cdots,x_M)\in \mathbb{F}_{q}^{M}$ is encoded by the outer MDS code, and outputs a codeword $y_\mathbf{f}=(y_1,y_2,\cdots,y_\theta)$. Second, every symbol of $y_\mathbf{f}$ is placed on $n$ storage nodes using a way defined by $\C$: place the symbol $y_i$ in the $j$th node if $i\in N_j$ in $\C$. The definition of FR code ensures that each node stores exactly $\alpha$ symbols, and each symbol is placed on exactly $\rho$ nodes.

A valid DRESS code should have the following  two properties.
First, when some node $j$ fails, it is possible to find a set of $d=\alpha$ other nodes, such that each node passing exactly one symbol is able to repair node $j$. The repair bandwidth $d$ is the same as the repair bandwidth of an MBR code.
Second, the stored file should be reconstructed from any set of $k$ nodes, which requires $\min_{|I|=k}|\cup_{i\in I} N_i|\geq M$
 due to the property of  the outer MDS code. Note that one can assume that $M=M(k)= \min_{|I|=k}|\cup_{i\in I} N_i|$ for a given DRESS code.

To maximize the file size and ensure correct reconstruction and repair, one can require that $|N_i\cap N_j|\leq 1$ for all $i\neq j$ \cite{rashmi2009explicit}. Let $A(n,k,\alpha,\rho)$ be the  maximum file size $M ( k )$ among all $[(\theta, M), k,(n, \alpha, \rho)]$ DRESS codes, which  indeed only depends on the inner FR code. Two upper bounds on  $A(n,k,\alpha,\rho)$ were given in \cite{el2010fractional},
\begin{equation}\label{ub1}
  A(n,k,\alpha,\rho)\leq \left\lfloor \frac{n\alpha}{\rho}\left( 1-\frac{{n-\rho \choose k}}{{n \choose k}} \right)\right\rfloor \text{ and }
\end{equation}
\begin{equation}\label{ub2}
    A(n,k,\alpha,\rho)\leq  \varphi(k), \text{ where }\varphi(1)=\alpha, \varphi(k+1)= \varphi(k)+\alpha-\left\lceil \frac{\rho \varphi(k)-k\alpha}{n-k}\right\rceil.
\end{equation}
An FR code is called \emph{$k$-optimal} if $\min_{|I|=k}|\cup_{i\in I} N_i|=A(n,k,\alpha,\rho)$ for a given $k$, and is \emph{optimal} if it is $k$-optimal for all $k\leq \alpha$.

Silberstein and  Etzion \cite{silberstein2015optimal}   constructed two kinds of optimal FR codes with $\rho=2$, one is based on  Tur\'an graphs and the other is based on graphs with large girth. For $\rho>2$, they showed that transversal designs and generalized polygons produce optimal FR codes.

\subsection{Set Systems and Graphs }
\indent For a finite set $V$ of points, let ${V\choose r}$ denote the set of all $r$-subsets of $V$. The pair ${\cal S}=(V,E)$ is called an \emph{$r$-uniform set system} if $E\subset {V\choose r}$. The elements of $E$ are called \emph{blocks}. The  \emph{order} of $\cal S$ is the number of points $|V|$, and the \emph{size} of $\cal S$ is the number of blocks $|E|$. Such a pair $(V,E)$ is also known as a graph if $r=2$, where $V$ and $E$ are commonly referred to vertices and edges, respectively. When $r>2$, $\cal S$ is known as an \emph{$r$-uniform hypergraph} (or \emph{$r$-graph}), where elements of $E$ are referred to hyperedges. A set system is called \emph{linear} if any two blocks intersect on at most one common point. The \emph{$2$-shadow} of ${\cal S}=(V,E)$ is a $2$-uniform set system denoted by $\partial_2 {\cal S}=(V',E')$, where $V'=V,E'=\{\{a,b\}:\{a,b\}\subseteq B, B\in E\}$. Clearly, if ${\cal S}$ is a $2$-uniform set system then $\partial_2 {\cal S} = {\cal S}.$

\par Two points $x,y\in V$ are \emph{adjacent} in $\cal S$, if there exists a block $e\in E$, such that $\{x,y\}\subset e$. A point $x$ is \emph{incident} with a block $e\in E$, if $x\in e$. The \emph{degree} of  $x$ is the number of blocks incident with $x$, denoted by $d(x)$. A set system is said to be \emph{$d$-regular} if $d(x)= d$ for all $ x \in V$, where $d$ is a positive integer.
The \emph{incidence matrix} $I({\cal S})$ of a set system ${\cal S}=(V,E)$ is a binary  $|V|\times |E|$ matrix with rows and columns indexed by $V$ and $E$, respectively, such that $I({\cal S})_{i,e}=1$ if and only if $ i\in e$. The \emph{line graph} $L({\cal S})$ of a set system ${\cal S}=(V,E)$, is a multi-edge graph $(V',E')$, where $V'=E$, and  the number of edges between two blocks  $e,e'\in E$ is $|e\cap e'|$. Note that, when $\cal S$ is linear, the line graph $L({\cal S})$  is a simple graph.  The \emph{dual set system} ${\cal S}^*$ of ${\cal S}$ is a set system with $I({\cal S}^*)=I({\cal S})^T.$

Now we give some definitions commonly used in graph theory.  The set system $\left(V,{V\choose 2}\right)$ is called a \emph{complete graph}, denoted by $K_n$ if $|V|=n$.  A graph is called an \emph{$r$-partite graph} if its vertices can be partitioned into $r$-parts, such that two vertices are adjacent only when they belong to different parts. It is further called \emph{complete $r$-partite} if every two vertices from different parts are adjacent. If a complete $r$-partite graph has parts of size $m_i$, $i\in [r]$, then we denote it by $K_{m_1,m_2,\ldots,m_r}$. A \emph{Tur\'an graph} $T(n,r)$ is an $n$-vertex complete $r$-partite graph, such that all parts are of size either $\lceil\frac{n}{r}\rceil$ or $\lfloor\frac{n}{r}\rfloor$.

The \emph{adjacency matrix} $A(G)$ of a graph $G=(V,E)$ is a $|V|\times |V|$ matrix whose rows and columns are indexed by $V$, such that  $A(G)_{i,j}=1$ if $\{i,j\}\in E$ and $0$ else. The  \emph{neighborhood} of a vertex $x$, denoted by $N(x)$, consists of all vertices $y$ that are adjacent to $x$. A \emph{cycle} in a graph $G$ is a connected $2$-regular subgraph of $G$. Denote ${\cal C}_n=(v_1,v_2,\ldots,v_n)$ a cycle with edges $v_i\sim v_{i+1}$, $i\in[n-1]$ and $v_n\sim v_1$. The \emph{girth} of a graph is the length of its shortest cycle. An \emph{independent set} of $G=(V,E)$ is a set of pairwise nonadjacent vertices. A \emph{perfect matching} of $G$ is a set of disjoint edges that cover all vertices. A graph $G$ is said to be \emph{$1$-factorable} if  $E$ can be partitioned into perfect matchings.

%
%

\subsection{FR codes based on set systems}\label{s:FRsyst}
Given an $(n,\alpha,\rho)$ FR code $\mathcal {C}$, the \emph{incidence matrix} $I(\mathcal{C})$ is an $n\times \theta$ binary matrix with $n\alpha=\theta \rho$, where rows are indexed by the nodes of the FR code,  columns are indexed by the symbols of  outer MDS codeword, and the entry $I(\mathcal{C})_{i,j}$ is defined as follows:
\begin{equation*}
I(\mathcal{C})_{i,j}=\left\{
\begin{array}{ll}{1} &{\text { if node } i \text{ contains symbol } j, } \\
{0}                  &{\text { otherwise}. }
\end{array}
\right.
\end{equation*}
Note that each row of $I(\mathcal{C})$ has exactly $\alpha$ ones and each column has exactly $\rho$ ones.  It is easy to see that a $\rho$-uniform $\alpha$-regular set system $\cal S$ of order $n$ gives an  $(n,\alpha,\rho)$ FR code $\cal C$ such that $I(\mathcal{C})=I(\mathcal{S})$. Since the transpose of $I(\mathcal{S})$ can be viewed as the incidence matrix of the dual of $\cal S$, which is an $\alpha$-uniform $\rho$-regular set system of order $n\alpha/\rho$, thus it also gives an $(n\alpha/\rho,\rho,\alpha)$ FR code $\cal C'$ such that $I(\mathcal{C'})=I(\mathcal{S})^T$.

\section{A new model of access-balancing FR codes}\label{s:3}

 Given a regular uniform set system, it can build an FR code as in Section~\ref{s:FRsyst}, where a node indexed by $x$ stores the content consisting of the indices of blocks containing the point $x$. Here, the indices of the blocks can be viewed as chunks of information, that is, symbols of the outer MDS codeword $(y_1,y_2,\cdots ,y_{\theta})$. The labels of chunks $y_1,y_2,\cdots ,y_{\theta}$ are directly proportional to their popularities and hence their access frequencies. It has been shown that the access frequencies of chunks experimentally obey Zipf law \cite{breslau1999web}, where the $i$-th most popular chunk has access frequency $1/i^\beta$ for some $\beta > 0$. To simplify our model, we assume that the labels of $y_1,y_2,\cdots ,y_{\theta}$ are consecutive integers in $[\theta]$, and discuss the case for Zipf law in the section of  Conclusion. The overall popularity of a node amounts to the sum of the labels of  chunks stored on the node. To make sure the access request as even as possible among all nodes, the authors in \cite{maxminsum} proposed  a chunk placement strategy, by what they referred to MaxMinSum placement. The MaxMinSum problem was stated as a point labeling problem for the dual design, which maximizes the minimum sum of block points for all blocks. They solved this problem for Steiner triple systems and their dual \cite{maxminsum}. Here, we restate this problem for  FR codes directly, i.e., a block labeling problem.

 \begin{problem}\cite{maxminsum}\label{maxminsum} Given a $\rho$-uniform $\alpha$-regular set system ${\cal S}=(V, E)$ of order $n$ with $V=\{v_1,v_2,\cdots ,v_n\}$ and $E=\{e_1,e_2,\cdots ,e_\theta\}$, where $\theta=n\alpha/\rho$, the problem of constructing a \emph{MaxMinSum} $(n,\alpha,\rho)$ FR code from $\cal S$  is equivalent to finding a labeling of blocks in $E$, i.e., a bijection $\sigma$ from $E$ to $[\theta]$,  such that the \emph{access-minsum} \[MinSum({\cal S}_\sigma):=\|I({\cal S})({\sigma(e_1)},{\sigma(e_2)},\ldots ,{\sigma(e_{\theta})})^T\|_{\mathbb{L}^{min}}\] is maximized.  Here, $\|\textbf{x}\|_{\mathbb{L}^{min}}:=\min \{x_i\}$ for $\textbf{x}=(x_1,x_2,\ldots,x_n)^T$.
\end{problem}

In \cite{maxminsum}, the authors also provided MinMaxSum model, which minimizes the \emph{access-maxsum} \[MaxSum({\cal S}_\sigma):=\|I({\cal S})({\sigma(e_1)},{\sigma(e_2)},\ldots ,{\sigma(e_{\theta})})^T\|_{\mathbb{L}^{max}},\] where  $\|\textbf{x}\|_{\mathbb{L}^{max}}:=\max \{x_i\}$ for $\textbf{x}=(x_1,x_2,\ldots,x_n)^T$. However, these two models can not be optimized simultaneously in general. For example,
let ${\cal C}$  be the FR code based on $K_4$ with nodes $\{1,2,3,4\}$ and six symbols,
there are two labelings $\sigma_1$ and $\sigma_2$ for ${\cal C}$ as follows:
\[
\begin{array}{ll} \sigma_1:\\
&\sigma_1(12)=3, \sigma_1(13)=1,\sigma_1(14)=6,\\
&\sigma_1(23)=5, \sigma_1(24)=2,\sigma_1(34)=4,\\
&MinSum({\cal S}_{\sigma_1})=\|(10,10,10,12)^T\|_{\mathbb{L}^{min}}=10,\\
&MaxSum({\cal S}_{\sigma_1})=\|(10,10,10,12)^T\|_{\mathbb{L}^{max}}=12.\\
\sigma_2:\\
&\sigma_2(12)=3, \sigma_2(13)=1,\sigma_2(14)=5,\\
&\sigma_2(23)=6, \sigma_2(24)=2,\sigma_2(34)=4,\\
&MinSum({\cal S}_{\sigma_2})=\|(9,11,11,11)^T\|_{\mathbb{L}^{min}}=9,\\
&MaxSum({\cal S}_{\sigma_2})=\|(9,11,11,11)^T\|_{\mathbb{L}^{max}}=11.\\
   \end{array}
\]
By easy computation, we see that $\sigma_1$ achieves the maximum access-minsum, while  $\sigma_2$ achieves the minimum access-maxsum. This means $\sigma_2$ is not an optimal labeling under the MaxMinSum model.  Now we look at the variance of popularities among all nodes, i.e., the value $\sum_{i=1}^\theta (p_i-\bar{p})^2$, where $p_i$ is the overall popularity of the  $i$th node, and $\bar{p}$  is the average popularity of all nodes. It is clear that $I({\cal S})({\sigma(e_1)},{\sigma(e_2)},\ldots ,{\sigma(e_{\theta})})^T=(p_1,p_2,\ldots ,p_\theta)^T$. Then
both $\sigma_1$ and $\sigma_2$ in the above example provide the same level of access-balancing FR codes when minimizing the variance of popularities. That is, the MaxMinSum model does not capture all good labelings with access-balancing property.

%

Let's look at another example of the FR code based on $K_8$. By computer search, we find two optimal labelings $\sigma_1$ and $\sigma_2$   under the MaxMinSum model, but $\sigma_2$  is clearly better than $\sigma_1$  for the access-balancing property when you consider their variances.
\[
\begin{array}{ll}\sigma_1:\\
       &\sigma_1(12)=1, \sigma_1(13)=2,\sigma_1(14)=3,\sigma_1(15)=14, \sigma_1(16)=26,\sigma_1(17)=27,\sigma_1(18)=28,\\
        &\sigma_1(23)=4, \sigma_1(24)=10,\sigma_1(25)=17,\sigma_1(26)=21, \sigma_1(27)=23,\sigma_1(28)=25,\sigma_1(34)=24,\\
         &\sigma_1(35)=22,\sigma_1(36)=15, \sigma_1(37)=16,\sigma_1(38)=18,\sigma_1(45)=20,\sigma_1(46)=19, \sigma_1(47)=12,\\
           &\sigma_1(48)=13, \sigma_1(56)=8,\sigma_1(57)=11,\sigma_1(58)=9, \sigma_1(67)=7,\sigma_1(68)=5,\sigma_1(78)=6,\\
           &MinSum({\cal S}_{\sigma_1})=\|(101,101,101,101,101,101,102,104)^T\|_{\mathbb{L}^{min}}=101.\\
\sigma_2:\\
       &\sigma_2(12)=1, \sigma_2(13)=2,\sigma_2(14)=3,\sigma_2(15)=14, \sigma_2(16)=26,\sigma_2(17)=27,\sigma_2(18)=28,\\
        &\sigma_2(23)=4, \sigma_2(24)=10,\sigma_2(25)=17,\sigma_2(26)=21, \sigma_2(27)=23,\sigma_2(28)=25,\sigma_2(34)=24,\\
         &\sigma_2(35)=22,\sigma_2(36)=15, \sigma_2(37)=16,\sigma_2(38)=18,\sigma_2(45)=20,\sigma_2(46)=19, \sigma_2(47)=12,\\
           &\sigma_2(48)=13, \sigma_2(56)=9,\sigma_2(57)=11,\sigma_2(58)=8, \sigma_2(67)=7,\sigma_2(68)=5,\sigma_2(78)=6,\\
           &MinSum({\cal S}_{\sigma_2})=\|(101,101,101,101,101,102,102,103)^T\|_{\mathbb{L}^{min}}=101.\\
   \end{array}
\]
Since the variance of popularities among all nodes is definitely a key factor that should be considered for the access-balancing property, we introduce a  new model to capture this property in next subsection.


\subsection{The MinVar model}
In this subsection, we present a new chunk placement strategy, which we call \emph{MinVar} placement. The MinVar placement is to minimize the variance of popularities among all nodes of the FR code by relabeling chunks. We formalize the problem as follows.

\begin{problem}\label{minvar} Given a $\rho$-uniform $\alpha$-regular set system ${\cal S}=(V, E)$ of order $n$ with $V=\{v_1,v_2,\cdots ,v_n\}$ and $E=\{e_1,e_2,\cdots ,e_\theta\}$, where $\theta=n\alpha/\rho$, the problem of constructing a \emph{MinVar} $(n,\alpha,\rho)$ FR code from $\cal S$  is equivalent to finding a labeling of blocks in $E$, i.e., a bijection $\sigma$ from $E$ to $[\theta]$,  such that the \emph{access-variance} \[Var({\cal S}_\sigma):=\|I({\cal S})({\sigma(e_1)},{\sigma(e_2)},\ldots ,{\sigma(e_{\theta})})^T-(\bar{a},\bar{a},\ldots ,\bar{a})^T\|_{\mathbb{L}^2}\] is minimized.  Here, $\bar{a}$ equals  the average popularity $\frac{\rho\theta(\theta+1)}{2n}=\frac{\alpha(\theta+1)}{2}$ and $\|\textbf{x}\|_{\mathbb{L}^2}:=\sum_{i=1}^{n}x_i^2$ for $\textbf{x}=(x_1,x_2,\ldots,x_n)^T$.
\end{problem}

Let $p_\sigma=(p_1,p_2,\ldots,p_{n}):=I({\cal S})({\sigma(e_1)},{\sigma(e_2)},\ldots ,{\sigma(e_{\theta})})^T$. By the definition of $I({\cal S})$, the $i$-th component  $p_i$ is the total popularity of node $i$, which equals to $\sum_{j:v_i\in e_j} \sigma(e_j)$. So the value of $Var({\cal S}_\sigma)/n$ can be viewed as the variance of the distribution of the popularities of $n$ nodes, which measures how far the popularity of each node is from the mean. Since the parameter $n$ is fixed, minimizing the  value of $Var({\cal S}_\sigma)$ can yield a  kind of evenly access request FR code.

Denote \[\textsf{MinVar}({\cal S})=\min_{\sigma}Var({\cal S}_\sigma).\] Then a \emph{MinVar FR code} is  an FR code based on $\cal S$ equipped with a labeling $\sigma$ such that  $Var({\cal S}_\sigma)=\textsf{MinVar}({\cal S})$. If $\cal S$ is a regular graph, then we find that $\textsf{MinVar}({\cal S})=0$ if and only if the graph is supermagic. We review some known results about magic labeling of graphs in the next subsection.

\subsection{Magic labeling}
The concept of magic labeling in graph theory was introduced by Sedl\'{a}\v{c}ek \cite{sedlavcek1963problem} in 1963, when considering
the notion of magic squares in number theory. After that,  Stewart
studied various problems to label the edges of a graph in \cite{stewart1966magic} and \cite{stewart1967supermagic}.
Given a connected graph $G=(V,E)$, and an injective mapping $\sigma$  from $E$ into positive integers, let
\[\sigma^*(v):=\sum_{e\in E:v\in e}\sigma(e).\] If $\sigma^*(v)=\lambda$ for all $v\in V$, then we say $\sigma$ is a \emph{magic labeling} of $G$ for an index $\lambda$. Further if $\{\sigma(e):e\in E\}$ consists of consecutive positive integers, then we say $\sigma$ is  \emph{supermagic}. A graph $G$ is supermagic (magic) whenever there exists a supermagic (magic) labeling of $G$.

 There is by now a considerable
number of papers published on magic and supermagic graphs, see for example \cite{doob1974generalizations,doob1978characterizations,jeurissen1988magic,sedlavcek1976magic,shiu2002construction,ivanvco2000supermagic,sun1994labeling}.
Regular supermagic graphs were extended to \emph{degree-magic} graphs if the set of labels is $[|E|]$, and $\sigma^*(v)=\deg(v)(1+|E|)/2$ for all $v\in V$ \cite{bezegova2012characterization,ivanvco2010extension}. Note that if $G$ is a regular graph, then $G$ is supermagic if and only if it is degree-magic \cite{ivanvco2010extension}. We refer the readers to \cite{gallian2018dynamic} for comprehensive references.

If $G$ is a supermagic (or degree-magic) regular graph, then the supermagic labeling $\sigma$ satisfies that $Var({G}_\sigma)=0$ by comparing the definition of  $Var({G}_\sigma)$ and degree-magic  labeling. In other words,  a supermagic regular graph can construct a MinVar FR code with zero access-variance.

 Ivan\v{c}o \cite{ivanvco2000supermagic}
 gave a characterization of all supermagic regular complete multipartite graphs, which we summarize as follows. Note that regular complete multipartite graphs are just regular Tur\'{a}n graphs.

 \begin{theorem}\label{tura} \cite{ivanvco2000supermagic}
    The  Tur\'{a}n graph  $T(n,r)$ with $r\mid n$ and $r\geq 2$, is supermagic if and only if one of the following conditions is satisfied:
     \begin{enumerate}[(1)]
  \item $n=r$, i.e., $K_n$,  with $n=2$, or $n\geq 6$ and $n\not\equiv0\pmod 4$;
  \item $n=2r\geq 6$, i.e., $T(2r,r)$ with $r\geq 3$;
  \item $n\geq 3r$,  except when $r \equiv 0 \pmod 4$ and $\frac{n}{r}$ is odd.
\end{enumerate}
    \end{theorem}

 \begin{theorem}\label{kopti} \cite{silberstein2015optimal}
    The  Tur\'{a}n graph  $T(n,r)$ with $r\mid n$ and $r\geq 2$,  gives a $k$-optimal $(n,\alpha,2)$ FR code for all $k\leq \alpha$, where $\alpha=(r-1)\frac{n}{r}$, hence gives an optimal $(n,\alpha,2)$ FR code.
    \end{theorem}

Combining Theorems~\ref{tura} and \ref{kopti}, we immediately get the following result.

\begin{corollary}\label{optMin}
Let $r\geq 2$, $r\mid n$ and $\alpha=(r-1)\frac{n}{r}$. There exists an optimal MinVar $(n,\alpha,2)$ FR code with zero access-variance if one of the following conditions is satisfied:
     \begin{enumerate}[(1)]
  \item $n=r=2$,  or $n=r\geq 6$ and $n\not\equiv0\pmod 4$;
  \item $n=2r\geq 6$;
  \item $n\geq 3r$,  except when $r \equiv 0 \pmod 4$ and $\frac{n}{r}$ is odd.
\end{enumerate}
\end{corollary}

\begin{question}\label{q:35}
  When $n=r\equiv 0 \pmod 4$, i.e., complete graphs with order divisible by four, or  $r \equiv 0 \pmod 4$ and $\frac{n}{r}$ is odd, what is $\textsf{MinVar}({\cal S})$ for ${\cal S}=T(n,r)$?
\end{question}

In \cite{silberstein2015optimal}, the authors showed that graphs with large girth can produce optimal FR codes.

 \begin{theorem}\label{girth} \cite{silberstein2015optimal}
    If there exists an $\alpha$-regular graph of girth $g$, then there exists a $k$-optimal $(n,\alpha,2)$ FR code for all $k\leq g-1$, where $n$ is the number of vertices of the graph. Further if $g\geq \alpha+1$, then the  FR code is optimal.
    \end{theorem}

\begin{question}\label{q:girth}
  For regular graphs with large girth, what is the minimum access-variance?
\end{question}

Note that Problem~\ref{minvar} can be viewed as an extension of the magic labeling problem  when a graph is not supermagic, which provides a reference to measure how far a labeling is from a supermagic labeling.

\subsection{An equivalent problem}
Problem~\ref{minvar} is to find  a block labeling of a set system to minimize the access-variance. Here, we present an equivalent problem, which looks for a good vertex labeling of the line graph.

For convenience, we always assume that ${\cal S}=(V, E)$ is a $\rho$-uniform $\alpha$-regular set system of order $n$ with $V=\{v_1,v_2,\cdots ,v_n\}$ and $E=\{e_1,e_2,\cdots ,e_\theta\}$, where $\theta=n\alpha/\rho$ and $I({\cal S})$ is the incidence matrix whose rows and columns are indexed by $V$ and $E$. Let ${\textbf{1}}_n$ be the all-$1$ column vector of length $n$. Then
\[I({\cal S})\cdot {\textbf{1}}_\theta=\alpha \cdot {\textbf{1}}_n, \text{ and }{\textbf{1}}_n^T \cdot I({\cal S}) =\rho\cdot {\textbf{1}}_\theta^T.\]

 Denote $M({\cal S}):={I({\cal S})}^T\cdot I({\cal S})$,  then  \[ M({\cal S})_{(i,j)}=\left\{
   \begin{array}{ll}{\rho} &{\text { if } i=j;} \\
   {|e_i\cap e_j|}         &{\text { if } i \neq j.}
   \end{array}
  \right.\]
If ${\cal S}$ is a linear set system, then $M({\cal S})=\rho I_{\theta} + A(L({\cal S}))$, where $I_{\theta}$ is an identity matrix of order $\theta,$
and $A(L({\cal S}))$ is the adjacency matrix of the line graph $L({\cal S})$. Clearly, $L({\cal S})$ is a $d$-regular graph with $\theta$ vertices, where $d=\rho(\alpha-1)$, i.e., $A(L({\cal S}))$ has $d$ ones in each row and each column. If ${\cal S}$ is not a linear set system, $A(L({\cal S}))$ is the adjacency matrix of a multi-graph $L({\cal S})$, where vertex $e_i$ is adjacent to vertex $e_j$ with exactly $|e_i\cap e_j|$ parallel edges.

\par \begin{lemma}\label{var}
Given a $\rho$-uniform  $\alpha$-regular linear set system  ${\cal S}=(V,E)$ of order $n$ and an edge labeling $\sigma$ with $\sigma(e_j)=i_j$, the access-variance \[Var({\cal S}_\sigma)=({i_1},{i_2},\cdots ,{i_\theta}) A(L({\cal S})) ({i_1},{i_2},\cdots ,{i_\theta})^T+c,\] where $c=c(\theta,\rho,\alpha)$ is a constant.
\end{lemma}
\begin{proof} Because $\sigma(e_j)=i_j$ is a bijection, we have  \[\sum\limits_{j=1}^{\theta}i_j=\sum\limits_{j=1}^{\theta}j=\frac{{\theta}^2+\theta}{2},~ \text{and } \sum\limits_{j=1}^{\theta}{i_j}^2=\sum\limits_{j=1}^{\theta}j^2=\frac{\theta(\theta+1)(2\theta+1)}{6}.\]
By $I({\cal S})\cdot {\textbf{1}}_\theta=\alpha \cdot {\textbf{1}}_n$, and ${{\textbf{1}}_\theta}^T\cdot A(L({\cal S}))=\rho(\alpha-1){\textbf{1}}_\theta$,  we have
  \begin{equation}\label{varmtx}
    \begin{aligned}
      Var({\cal S}_\sigma)&=\|I({\cal S})({i_1-\frac{\bar{a}}{\alpha}},{i_2-\frac{\bar{a}}{\alpha}},\cdots ,{i_\theta-\frac{\bar{a}}{\alpha}})^T\|_{\mathbb{L}^2}\\
            &=({i_1-\frac{\bar{a}}{\alpha}},{i_2-\frac{\bar{a}}{\alpha}},\cdots ,{i_\theta-\frac{\bar{a}}{\alpha}}) M({\cal S}) ({i_1-\frac{\bar{a}}{\alpha}},{i_2-\frac{\bar{a}}{\alpha}},\cdots ,{i_\theta-\frac{\bar{a}}{\alpha}})^T\\
            &\triangleq(i_1,i_2,\cdots ,i_\theta) M({\cal S}) (i_1,i_2,\cdots ,i_\theta)^T+c_1\\
            &\triangleq({i_1},{i_2},\cdots ,{i_\theta})A(L({\cal S}))({i_1},{i_2},\cdots ,{i_\theta})^T+c_1+c_2\\
            &\triangleq\sum_{e_k\cap e_l\neq \emptyset}i_ki_l+c,
    \end{aligned}
  \end{equation}
  where
  \begin{equation}\label{varmtx}
    \begin{aligned}
      c_1&=(-\frac{\bar{a}}{\alpha}\mathbf{1_\theta}^T) M({\cal S}) (-\frac{\bar{a}}{\alpha}\mathbf{1_\theta})-2(i_1,i_2,\cdots ,i_\theta) M({\cal S})(\frac{\bar{a}}{\alpha}\mathbf{1_\theta})\\
         &=\frac{\bar{a}^2}{\alpha^2} \mathbf{1_\theta}^T \rho\alpha\mathbf{1_\theta}-2\frac{\bar{a}}{\alpha}(i_1,i_2,\cdots ,i_\theta) \rho\alpha\mathbf{1_\theta} ~~~(\text{Since }M({\cal S})\mathbf{1_\theta}=\rho\alpha\mathbf{1_\theta})\\
         &=\frac{\bar{a}^2}{\alpha}\rho\theta-2\bar{a}\rho\sum_{j=1}^\theta j=\bar{a}\rho\theta\left(\frac{\bar{a}}{\alpha}-(\theta+1)\right)\\
         &=-\frac{\bar{a}\rho\theta(\theta+1)}{2},  ~~~\left(\text{Since } \frac{\bar{a}}{\alpha}=\frac{\theta+1}{2}\right)
    \end{aligned}
  \end{equation}
%
 and $c_2=\rho \sum\limits_{i=1}^\theta i^2$. Then $c=c_1+c_2=\frac{\rho\theta(\theta+1)(2\theta+1)}{6}-\frac{\rho \alpha \theta(\theta+1)^2}{4}.$
\end{proof}
\par From Lemma \ref{var}, we only need to consider the line graph of ${\cal S}$ in Problem~\ref{minvar}. So it is natural to propose the following optimization problem about graph labeling of vertices, which is equivalent to Problem~\ref{minvar} if we let $G=L(\cal S)$.
\begin{problem}\label{graphpro}
  Given a $d$-regular graph $G$ with $\theta$ vertices, $v_1,v_2,\cdots,v_{\theta}$,  find a weight function $f:V(G)\rightarrow [\theta]$, which is a bijection, such that ${\cal M}(f):=\sum_{v_i\sim v_j}f(v_i)f(v_j)$ is minimized. Denote ${\cal M}(G)=\min_f{\cal M}(f)$, and call it the \textsf{MinPS} of $G$ (which stands for the minimum product sum). Note that each edge in ${\cal M}(f)$ is computed only \emph{once} in the summation.
\end{problem}

 Now we give the motivation of why we study Problem \ref{graphpro} for union of complete or Tur\'{a}n graphs, cycles in Section~\ref{s:4}. As in Questions~\ref{q:35} and \ref{q:girth}, we are interested in the value of $\textsf{MinVar}({\cal S})$ when ${\cal S}$ is a Tur\'{a}n  graph or a graph with large girth. Alternatively, we can estimate the corresponding \textsf{MinPS} of $G=L(\cal S)$ in Problem \ref{graphpro}. For example, when ${\cal S}$ is a cycle (a graph with large girth), $L(\cal S)$ is also a cycle, whose \textsf{MinPS} can be determined as in Section~\ref{s:cycle}. However, for general Tur\'{a}n  graphs $T(n,r)$, it is not easy to determine the \textsf{MinPS} of $L(\cal S)$, even when ${\cal S}=T(n,n)$ with $n\equiv 0 \pmod 4$, i.e., the complete graph $K_{n}$. For this case, we will decompose $K_{n}$ into a Tur\'{a}n  graph $T(n,n/4)$, which is supermagic, and a union of $n/4$ disjoint $K_4$'s, whose line graph is a union of $n/4$ disjoint Tur\'{a}n  graphs $T(6,3)$. When determining the \textsf{MinPS} for the union of Tur\'{a}n graphs, we need to determine it for the union of complete  graphs as in Lemma~\ref{mtnr}. Besides, our results for these graphs can give solutions to Problem~\ref{minvar} for some FR codes, see the examples given in the beginning of each subsection of Section IV.

Before closing this section, we give a general upper bound on the \textsf{MinPS} of $G$, which will be improved in Section~\ref{s:5} for $G$ being the line graph of a complete graph.
 \begin{lemma}\label{gebd}
   Given a $d$-regular graph $G$ with ${\theta}$ vertices,  ${\cal{M}}(G)\leq \frac{d(3{\theta}+2){\theta}({\theta}+1)}{24}.$
 \end{lemma}

\begin{proof}
  We compute the average value of ${\cal{M}}(f)$,
  \begin{equation*}
    \begin{aligned}
     \sum_{\sigma \in S_{\theta}}{\cal{M}}(f_\sigma)&=\sum_{\sigma\in S_{\theta}}\sum_{v_i\sim v_j}f_\sigma(v_i)f_\sigma(v_j)=\sum_{v_i\sim v_j}\sum_{\sigma\in S_{\theta}}f_\sigma(v_i)f_\sigma(v_j)\\
                                             &=\sum_{v_i\sim v_j}({\theta}-2)!\sum_{1\leq a\neq b\leq {\theta}}ab=\sum_{v_i\sim v_j}({\theta}-2)!\left(\left(\sum_{i=1}^{\theta}i\right)^2-\sum_{i=1}^{{\theta}}i^2\right)\\
                                             &=\frac{d(3{\theta}+2){\theta}^2({\theta}^2-1)}{24}({\theta}-2)!.
        \end{aligned}
  \end{equation*}
 Hence, ${\cal{M}}(G)\leq \frac{d(3{\theta}+2){\theta}({\theta}+1)}{24}.$
\end{proof}

\section{Determination of \textsf{MinPS} in Problem~\ref{graphpro}}\label{s:4}

In this section, we focus on solving Problem~\ref{graphpro} for several classes of graphs, such as disjoint union of complete graphs, Tur\'{a}n graphs and cycles.  For convenience, let $[a,b]$ denote the set of integers $\{a,a+1,\ldots,b\}$ for any integers $a\leq b$, and $[1,b]$ is abbreviated to $[b]$.

Before solving Problem~\ref{graphpro}, we illustrate the relation between the MinVar problem and the MinPS problem for  linear set systems and their dual  in  Figure \ref{Fig:relationship}.

\begin{figure}[h]
\center{
 \begin{tikzpicture}
 \draw[rounded corners] (-5,1) rectangle (-1,3);
 \draw (-3,2.5) node []{MinVar problem for };
 \draw (-3,2) node []{a set system ${\cal S}$};
 \draw (-3,1.5) node []{in Problem \ref{minvar}};

 \draw[rounded corners] (2,1) rectangle (6,3);
  \draw (4,2.5) node []{MinVar problem for};
 \draw (4,2) node []{the dual system ${\cal S}^*$};
 \draw (4,1.5) node []{in Problem \ref{minvar}};

 \draw[rounded corners] (2,-3.5) rectangle (6,-1.5);
 \draw (4,-2) node []{MinPS problem for};
 \draw (4,-2.5) node []{$L({\cal S}^*)=\partial_2 {\cal S}$};
 \draw (4,-3) node []{in Problem \ref{graphpro}};

 \draw[rounded corners] (-5,-3.5) rectangle (-1,-1.5);
 \draw (-3,-2) node []{MinPS problem for};
 \draw (-3,-2.5) node []{$L({\cal S})$};
 \draw (-3,-3) node []{in Problem \ref{graphpro}};

 \draw[->] (-0.8,-0.3)--(1.8,-0.3);
 \draw[<-] (-0.8,-0.4)--(1.8,-0.4);
 \draw (-3,0.8)--(-3,-1.3);
 \draw (4,0.8)--(4,-1.3);

 \draw (0.35,0) node []{Dual};
 \draw (-4,-0.4) node []{Equivalent};
 \draw (5,-0.4) node []{Equivalent};

 \end{tikzpicture}
 \caption{Relationship of the MinVar  and MinPS problems for linear set systems and their dual.  Any two blocks in ${\cal S^*}$ are incident if and only if the corresponding points in ${\cal S}$ are in the same block, thereby $L({\cal S}^*)=\partial_2 {\cal S}$.}\label{Fig:relationship}}
\end{figure}
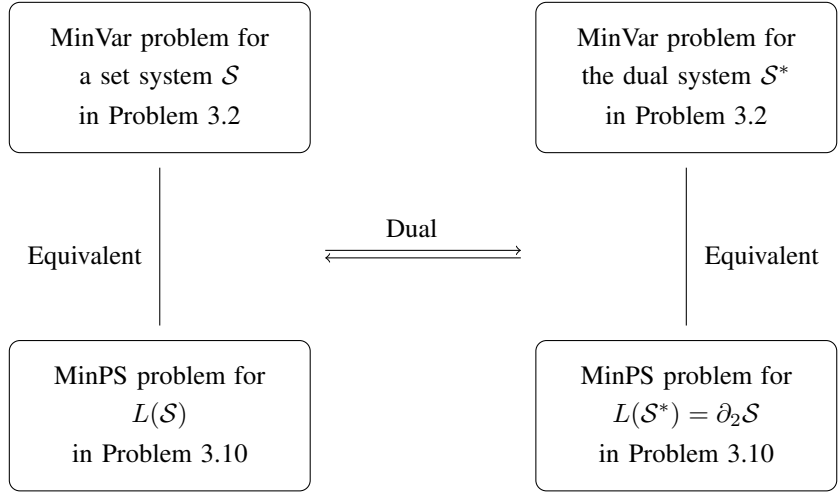

\subsection{Union of complete graphs}
]Let $mK_r$ be a graph that is a disjoint union of $m$ copies of $K_r$. In this subsection, we consider the \textsf{MinPS} of $mK_r$, whose value will be used in Lemma~\ref{mtnr} for a union of Tur\'{a}n graphs $mT(n,r)$.

When $m=1$, it is just a complete graph $K_r$. It is easy to see that the value ${\cal M}(f)$ is a constant for all labelings $f$, and ${\cal M}(K_r)=\sum_{1\leq i< j\leq r}ij$. FR codes $\cal S$ with $L({\cal S})$ being complete graphs exist, for example when $\cal S$ is a $2$-$(q^2+q+1,q+1,1)$ design (i.e., symmetric design \cite{colbourn2010crc}), in which every pair of blocks intersect in exactly one point.


  \begin{lemma}\label{mk}
 For $m\geq 2$, the \textsf{MinPS} value ${\cal M}(mK_r)$ can be achieved by the labeling satisfying that the label sums for each copy are as equal as possible.
 \end{lemma}
 \begin{proof}
  Given a vertex labeling $f$ of $mK_r$, let $V_i$ be the set of labels of the $i$th copy of $K_r$, then $|V_i|=r$ and $\cup_{i=1}^{m}V_i=[mr]$.
  Then
  \begin{equation*}
    \begin{aligned}
    {\cal M}(f)&=\sum_{i=1}^{m} \sum_{u<v\in V_i}uv=\sum_{i=1}^{m} \frac{1}{2} \left(\left(\sum_{u\in V_i}u\right)^2-\sum_{u\in V_i}u^2\right)\\
    &=\frac{1}{2} \sum_{i=1}^{m}\left(\sum_{u\in V_i}u\right)^2-\frac{1}{2} \sum_{i=1}^{mr}i^2\\
    &\geq \frac{m}{2}\left(\frac{\sum_{i=1}^{mr}i}{m}\right)^2-\frac{mr(mr+1)(2mr+1)}{12}\\
    &=\frac{mr^2(mr+1)^2}{8}-\frac{mr(mr+1)(2mr+1)}{12}.\\
        \end{aligned}
  \end{equation*}
  The inequality holds when $\sum_{u\in V_i}u$ is the same for all $i\in [m]$.  This can be achieved except when $r$ is odd and $m$ is even (in this case, $\frac{\sum_{i=1}^{mr}i}{m}$ is not an integer), for which
 the minimum value  can be achieved if $\sum_{u\in V_i}u$ are almost the same for all $i\in [m]$, i.e., pairwise difference is at most one.
 \end{proof}
 \begin{Remark}\label{Mmkr}
    The value of ${\cal M}(mK_r)$ in Lemma~\ref{mk} can be achieved by the following constructions.
    \begin{enumerate}[(1)]
      \item When $r$ is  even, for each $i\in [m]$, let
      \[V_i=\left[\frac{(i-1)r}{2}+1,\frac{ir}{2}\right]\bigcup\left[mr+1-\frac{ir}{2},mr-\frac{(i-1)r}{2}\right]:=V_i^{(r)}.\]
Then $\sum_{u\in V_i}u=\frac{(rm+1)r}{2}$ for all $i$, and ${\cal M}(mK_r)=\frac{mr^2(mr+1)^2}{8}-\frac{mr(mr+1)(2mr+1)}{12}$ .
      \item When $m$ and $r>1$ are both odd, let
       \[V_i=V_i^{(r-3)}\bigcup \left\{(r-3)m+i,\frac{(2r-3)m-1}{2}+i,rm+2-2i\right\},i\in\left [\frac{m+1}{2}\right]\]
       and
       \[V_i=V_i^{(r-3)}\bigcup \left\{(r-3)m+i,\frac{(2r-5)m-1}{2}+i,(r+1)m+2-2i\right\},i\in \left[\frac{m+3}{2}, m\right].\]
Then   $\sum_{u\in V_i}u=\frac{(rm+1)r}{2}$ for all $i$, and ${\cal M}(mK_r)=\frac{mr^2(mr+1)^2}{8}-\frac{mr(mr+1)(2mr+1)}{12}.$
        \item When $r>1$ is  odd and $m$ is even, let
       \[V_i=V_i^{(r-3)}\bigcup \left\{(r-3)m+i,\frac{(2r-3)m}{2}+i,rm+2-2i\right\}, i\in \left[\frac{m}{2}\right]\]
and
       \[V_i=V_i^{(r-3)}\bigcup \left\{(r-3)m+i,\frac{(2r-5)m}{2}+i,(r+1)m+1-2i\right\}, i\in\left[\frac{m}{2}+1,m\right].\]
Then    \begin{equation*}
\sum_{u\in V_i}u=\begin{dcases}
\frac{(rm+1)r+1}{2}, &i\in \left[\frac{m}{2}\right];\\
\frac{(rm+1)r-1}{2}, &i\in\left[\frac{m}{2}+1,m\right].
\end{dcases}
\end{equation*}
Hence, ${\cal M}(mK_r)=\frac{mr^2(mr+1)^2+m}{8}-\frac{mr(mr+1)(2mr+1)}{12}$.
    \end{enumerate}
 \end{Remark}

\subsection{Union of Tur\'{a}n graphs}

Next, we determine the \textsf{MinPS} of  a Tur\'{a}n graph  $T(n,r)$ with $r\mid n$.  FR codes  based on the dual transversal designs TD$(r,n/r)$ \cite{colbourn2010crc} have the line graph $T(n,r)$ for any $r\mid n$. We first provide  a simple but very useful remark below, which will be frequently used in our proofs.

\begin{Remark}\label{rem2}  Suppose that $f$ is a weight function of $G$ achieving the \textsf{MinPS}. For every vertex  $v$, denote  $f(N(v)):=\sum_{ u\in N(v)}f(u)$.  If $v_i$ and $v_j$ are not adjacent and $f(v_i)< f(v_j)$, then we claim that $f(N(v_i))\geq f(N(v_j))$. In fact, if $f(N(v_i))<f(N(v_j))$, then we switch the weight value between $v_i$ and $v_j$, and obtain a new weight function $f'$, for which  $f'(N(v_i))=f(N(v_i)),f'(N(v_j))=f(N(v_j))$ and $f'(v_i)=f(v_j),f'(v_j)=f(v_i).$ But
\begin{equation*}
    \begin{aligned}
    {\cal M}(f')-{\cal M}(f)&=f'(v_i)f'(N(v_i))+f'(v_j)f'(N(v_j))-f(v_i)f(N(v_i))-f(v_j)f(N(v_j))\\
               &=f(v_i)(f(N(v_j))-f(N(v_i)))+f(v_j)(f(N(v_i))-f(N(v_j)))\\
               &=(f(N(v_i))-f(N(v_j)))(f(v_j)-f(v_i))<0,
    \end{aligned}
  \end{equation*}
which is a contradiction to the fact that  ${\cal M}(f)={\cal M}(G).$
\end{Remark}

\begin{lemma}\label{multip} Let $r\mid n$ and $r\geq 2$. Then the \textsf{MinPS} of $T(n,r)$
\[{\cal M}(T(n,r))=\sum_{1\leq j < j'\leq r}\left(\sum_{k=(j-1)n/r+1}^{jn/r}k\right) \left( \sum_{l=(j'-1)n/r+1}^{j'n/r}l\right).\]
\end{lemma}
\begin{proof} Let $m=n/r\geq 2$ (when $m=1$, the graph is a complete graph, which is trivial). Suppose that $G=(V,E)$ is a Tur\'{a}n graph  $T(n,r)$, with $V=V_1\cup V_2\cup \cdots \cup V_r$ and $|V_i|=m$.  Given a weight function $f$, denote $f(V_i)=\sum_{v\in V_i}f(v)$ and $f(V)=\sum_{v\in V}f(v)$, then ${\cal M}(f)=\sum_{1\leq i < j\leq r}f(V_i)f(V_j)$. Without loss of generality, we assume that $f(V_1)\leq f(V_2)\leq \cdots \leq f(V_r)$.

 We claim that the weight function which minimizes ${\cal M}(f)$ has the property that $f(V_i)=\sum_{j=(i-1)m+1}^{im}j$, $i\in [r]$, that is, the labels of vertices in $V_i$ are $(i-1)m+1,(i-1)m+2,\ldots,im$.

 We prove it by contradiction.  Assume that for some $i<i'$,  there exists a label $x\in V_i$ and a label $y\in V_{i'}$ satisfying $x>y$. By switching the labels $x$ and $y$, we get a new weight function $f'$, for which
   \begin{equation}
    \begin{aligned}
       {\cal M}(f')-{\cal M}(f)&=\sum_{1\leq i < j\leq r}f'(V_i)f'(V_j)-\sum_{1\leq i < j\leq r}f(V_i)f(V_j)\\
                               &=(f(V_i)-x+y)(f(V_{i'})-y+x)-f(V_i)f(V_{i'})\\
                               &=(f(V_i)-f(V_{i'}))(x-y)-(x-y)^2.
    \end{aligned}
  \end{equation}
  Since $f(V_{i'})\geq f(V_i)$ and $x>y$, we have ${\cal M}(f')<{\cal M}(f)$, a contradiction.
\end{proof}

Now we consider the graph $mT(n,r)$ with $r\mid n$, which is a disjoint union of copies of $T(n,r)$. The \textsf{MinPS} of $mT(n,r)$ will be used in the estimation of $\textsf{MinVar}(K_{4r})$ in Lemma~\ref{k4r}.

\begin{lemma}\label{mtnr}
The \textsf{MinPS} of $mT(n,r)$  with $r\mid n$ can be determined by the \textsf{MinPS} of $mK_r$.
\end{lemma}
\begin{proof}
Let $l=n/r\geq 2$ and $V_{ij}$ be the set of labels of vertices from the $i$th part of the $j$th copy of $T(n,r)$, $i\in [r]$ and $j\in [m]$. For any labeling $f$, we have $|V_{ij}|=l$ and $\bigcup_{i\in [r],j\in[m]}V_{ij}=[mn]$. Suppose that $f$ minimizes ${\cal M}(f)$. Let $V_j=\cup_{i\in [r]}V_{ij}$. By the proof of Lemma~\ref{multip}, we know that for each $j$, $V_{1j}$ is the set of the smallest $l$ integers from $V_j$,  $V_{2j}$ is the set of the smallest $l$ integers from $V_j\setminus V_{1j}$, and so on. For convenience, we denote this property by $\cal P$.

We  claim that for any two different sets $V_{ij}$ and $V_{i'j'}$, all integers in $V_{ij}$ are either smaller than each integer in $V_{i'j'}$,  or greater than each integer in $V_{i'j'}$. That is to say, each $V_{ij}$ must be exactly a set  $U_t=[l(t-1)+1,lt]$ for some $t\in[mr]$. We prove the claim by contradiction. Suppose that there exist two sets $V_{ij}$ and $V_{i'j'}$, and integers $x,y\in V_{ij}$ and $z,w\in V_{i'j'}$, such that $x>z$ and $y<w$. By property $\cal P$, we have $j\neq j'$. Let $F_j=\sum_{s\neq i} V_{sj}$ and $F_{j'}=\sum_{a\neq i'}V_{aj'}$. Suppose that $F_j\geq F_{j'}$. Let $f'$ be a new labeling by switching $x$ and $z$. Then
\begin{equation*}
    \begin{aligned}
     {\cal M}(f')-{\cal M}(f)&=(zF_j+xF_{j'})-(xF_j+zF_{j'})\\
               &=(z-x)(F_j-F_{j'})\leq 0.
    \end{aligned}
  \end{equation*}
Thus, by switching $x$ and $z$, the value of ${\cal M}(f)$ does not increase. Continuing this operation, we can have all integers in $V_{ij}$ are smaller than those in $V_{i'j'}$. Repeating this step for any such pair $V_{ij}$ and $V_{i'j'}$, eventually each $V_{ij}$ becomes some $U_t$. 

By the above claim, we can assume that  each $V_{ij}$ is a set $U_t$ for some $t$. Now we compute ${\cal M}(f)$. Assume that $V_{ij}=U_t$ and $V_{i'j}=U_{t'}$ in the following equation.
\begin{equation*}
    \begin{aligned}
 {\cal M}(f)&=\sum_{j\in[m]}\sum_{1\leq i < i'\leq r}\left (\sum_{u\in V_{ij}}u\right)\left(\sum_{u\in V_{i'j}}u\right)\\
 &=\sum_{j\in[m]}\sum_{1\leq i < i'\leq r} \frac{l(2lt-l+1)}{2}\times \frac{l(2lt'-l+1)}{2}\\
 &={r\choose 2}\frac{m l^2(1-l)^2}{4}+ l^3 (1-l) \sum_{j\in[m]}\sum_{1\leq i < i'\leq r}\frac{t+t'}{2}+l^4 \sum_{j\in[m]}\sum_{1\leq i < i'\leq r}tt'\\
 &={r\choose 2}\frac{m l^2(1-l)^2}{4}+ \frac{l^3-l^4}{4} rm(rm+1)(r-1)+l^4{\cal M}(\bar{f}),
    \end{aligned}
  \end{equation*}
  where  $\bar{f}$ is the induced vertex labeling for $mK_r$. By Remark \ref{Mmkr}, we can find the optimal labeling $\bar{f}$ with ${\cal M}(\bar{f})={\cal M}(mK_r)$, from which we can deduce the optimal labeling $f$ for $mT(n,r)$.
\end{proof}

\subsection{Cycles}\label{s:cycle}

Large cycles are graphs with large girth, whose line graphs are also themselves.
Now we consider the \textsf{MinPS} of  the cycle $C_\theta$.  Denote ${\cal M}_\theta={\cal M}({\cal C}_\theta)$.
\begin{lemma}\label{cycupp}For any $\theta\geq 3$, we have ${\cal M}_{\theta+2}\geq {\cal M}_\theta+\theta^2+4\theta+5$.
\end{lemma}
\begin{proof}
  We prove it by contradiction. Suppose that there exists a labeling $f$ of ${\cal C}_{\theta+2}$  such that ${\cal M}(f)$ is minimized and ${\cal M}(f)<{\cal M}_\theta+\theta^2+4\theta+5$. From $f$, if we can deduce a labeling $f'$ of ${\cal C}_\theta$ satisfying ${\cal M}(f')<{\cal M}_\theta$, then we are done. We split the proof into three cases.

  \textbf{Case 1}: Suppose that in ${\cal C}_{\theta+2}$, we have a segment with labels $\cdots x~\theta+2~1~y\cdots$. We operate this cycle in the  following two steps.
  \begin{enumerate}
    \item[(S1)]Delete the two vertices with labels $1$ and $\theta+2$, and connect the two vertices with labels $x$ and $y$. The value of ${\cal M}(f)$ becomes $M={\cal M}(f)-x(\theta+2)-(\theta+2)-y+xy.$
    \item[(S2)]Change each label $l$ by $l-1$ of vertices in the cycle of length $\theta$. Then we obtain a labeling $f'$ of  a cycle ${\cal C}_\theta$.
  \end{enumerate}  Now  we compute ${\cal M}(f')$. Note that in (S1), $M=\sum_{u\sim v}uv$, and each $u\in [2,\theta+1]$ appears twice in this summation. Since $uv=(u-1)(v-1)+u+v-1$, we have
          \[M=\sum_{u\sim v}(u-1)(v-1)+2\sum_{i=2}^{\theta+1}i-\theta={\cal M}(f')+\theta^2+2\theta.\]
          So after (S2), we have
             \begin{equation}
    \begin{aligned}
       {\cal M}(f')& ={\cal M}(f)-x(\theta+2)-(\theta+2)-y+xy-\theta^2-2\theta\\
       &<{\cal M}_\theta+\theta^2+4\theta+5-x(\theta+2)-(\theta+2)-y+xy-\theta^2-2\theta\\
       &={\cal M}_\theta+(x-1)(y-\theta-2)+1.\\
    \end{aligned}
  \end{equation}
Since  $y<\theta+2$ and $x\geq 2$, we have ${\cal M}(f')<{\cal M}_\theta$, which is a contradiction.

           \textbf{Case 2}: Suppose that in ${\cal C}_{\theta+2}$, we have a segment with labels $\cdots x~\theta+2~z~1~y\cdots$.
Delete the two vertices with labels $1$ and $\theta+2$, and connect $xz$ and  $yz$. The value of ${\cal M}(f)$ becomes \[M={\cal M}(f)-(x+z)(\theta+2)-(y+z)+xz+yz.\]
Applying (S2) in Case 1, we obtain a labeling $f'$ of  a cycle ${\cal C}_\theta$.
Similar to Case 1, we have
             \begin{equation}
    \begin{aligned}
       {\cal M}(f')&=M-\theta^2-2\theta\\
        &={\cal M}(f)-(x+z)(\theta+2)-(y+z)+xz+yz-\theta^2-2\theta\\
        &<{\cal M}_\theta+\theta^2+4\theta+5-(x+z)(\theta+2)-(y+z)+xz+yz-\theta^2-2\theta\\
        &={\cal M}_\theta+(x-2)(z-\theta-2)+(z-1)(y-1)-z\theta.\\
    \end{aligned}
  \end{equation}
Since  $y,z<\theta+2$ and $x\geq 2$, we have ${\cal M}(f')<{\cal M}_\theta$, which is a contradiction.

           \textbf{Case 3}: Suppose that in ${\cal C}_{\theta+2}$, we have a segment with labels $\cdots x~\theta+2~z\cdots w~1~y\cdots$. We claim that $z$ and $w$ are not adjacent. Otherwise, $f(N(z))=\theta+2+w$ and $f(N(1))=y+w$, thus $f(N(1))<f(N(z))$, which contradicts to Remark~\ref{rem2}. Similarly, $x$ and $y$ can not be adjacent. Suppose $w>y$ (the case when $w<y$ is similar), and a label $a$ is next to $w$ on the left, i.e., the segment is $\cdots a~w~1~y\cdots$ around $1$. We operate this cycle in the  following two steps.
  \begin{enumerate}
    \item[(T1)]Delete the two vertices with labels $1$ and $w$, and connect the two vertices with labels $a$ and $y$. The value of ${\cal M}(f)$ becomes $M={\cal M}(f)-aw-w-y+ay.$
    \item[(T2)]Change each label $l>w$ by $l-2$, and each label  $l<w$ by $l-1$ in the cycle of length $\theta$. Then we obtain a labeling $f'$ of  a cycle ${\cal C}_\theta$.
  \end{enumerate}
 Since $(u-2)(v-2)=uv-2(u+v)+4$, $(u-1)(v-2)=uv-2u-v+2$ and $(u-1)(v-1)=uv-(u+v)+1$, we have
            \begin{equation}
    \begin{aligned}
       {\cal M}(f')&=  M-\sum_{u\sim v: u,v> w}(2u+2v-4)-\sum_{u\sim v: u< w<v}(2u+v-2)-\sum_{u\sim v:u,v< w}(u+v-1)\\
         &= M-4\sum_{i=w+1}^{\theta+2}i-2\sum_{i=2}^{w-1}i+\sum_{u\sim v: u< w<v}(v-u+2)+\sum_{u\sim v:u,v> w}4+\sum_{u\sim v:u,v< w}1\\
         &\triangleq M-2(\theta+w+3)(\theta-w+2)-(w+1)(w-2)+S,
    \end{aligned}
  \end{equation}
where $S$ is the sum of the last three terms.   Now we compute ${\cal M}(f')-{\cal M}_\theta$, which is less than
\begin{equation}\label{maxdiff}
 T(w):=\theta^2+4\theta+5-aw-w-y+ay-2(\theta+w+3)(\theta-w+2)-(w+1)(w-2)+S.
\end{equation}

We claim that $T(w)<0$ for any $w$, thus ${\cal M}(f')<{\cal M}_\theta$, a contradiction. We prove it by upper bounding the value of $S$.
Note that when $u\sim v,$ and $u< w<v$, one has $v-u+2\geq 4.$ Hence, the more crossing edges there are between $[2,w-1]$ and  $[w+1,\theta+2]$, the bigger $S$ it is. Based on the value of $w$, we split into two cases.

If $w\leq \frac{\theta}{2}+2,$ then $\left|[2,w-1]\right|\leq |[w+1,\theta+2]|$. Therefore,
\begin{equation*}
    \begin{aligned}
       S&\leq 2\sum\limits_{i=\theta-w+5}^{\theta+2}i-2\sum\limits_{i=2}^{w-1}i+4(w-2)+4(\theta-2w+4)\\
       &=(2\theta-w+11)(w-2)-(w+1)(w-2)+4(\theta-2w+4)\\
       &=-2w^2+(2\theta+6)w-4,
    \end{aligned}
  \end{equation*}
  which is maximized when $w=\frac{\theta}{2}+2$. Hence,
  \begin{equation*}
    \begin{aligned}
      T(w) & \leq -w^2+(2\theta+8)w-(\theta+3)^2+a(y-w)-y \\
       & =a(y-w)-y-\frac{\theta^2}{4}+3<0.
    \end{aligned}
  \end{equation*}

If $w\geq \frac{\theta}{2}+2,$ then $|[2,w-1]|\geq |[w+1,\theta+2]|.$ Therefore,
\begin{equation*}
    \begin{aligned}
      S&\leq 2\sum\limits_{i=w+1}^{\theta+2}i-2\sum\limits_{i=2}^{\theta-w+3}i+4(\theta-w+2)+(2w-\theta-4)\\
      &=(\theta+w+7)(\theta-w+2)-(\theta-w+5)(\theta-w+2)+(2w-\theta-4)\\
      &=-2w^2+(2\theta+4)w+\theta,
    \end{aligned}
  \end{equation*}
  which is maximized when $w=\theta+1$. Then
  \begin{equation*}
    \begin{aligned}
      T(w) & \leq  -w^2+(2\theta+5)w-\theta^2-5\theta-5+(a-1)(y-w)\\
       & =(a-1)(y-w)-1< 0.
    \end{aligned}
  \end{equation*}\end{proof}

Next, we show that the  lower bound of the  \textsf{MinPS} of  $C_\theta$ in Lemma~\ref{cycupp} can be achieved.

\begin{lemma}\label{cycexa}
  For any $\theta\geq 3$, we have ${\cal M}_{\theta+2}= {\cal M}_\theta+\theta^2+4\theta+5$. Hence, ${\cal M}_{2k+1}=(4k^3+12k^2+14k+3)/3$ and ${\cal M}_{2k+2}=(4k^3+18k^2+29k+12)/3$ for each $k\geq 1$.
\end{lemma}
\begin{proof}
We prove a stronger statement that for each   $\theta\geq 3$, there exists a labeling $f$ for ${\cal C}_\theta$ achieving ${\cal M}_\theta$ such that the labels $1$ and $\theta$ are adjacent.

For $\theta=3$, ${\cal M}(f)$ is constant and ${\cal M}_{3}=11$, and two labels $1$ and $3$ are adjacent. For $\theta=4$, by Lemma~\ref{multip}, ${\cal M}_{4}=21$, and two labels $1$ and $4$ are adjacent. Assume that for any cycles of length no more than $\theta$, our statement is true, then we prove that it is true for $\theta+2$. By assumption, there exists a  labeling $f'$ with ${\cal M}(f')={\cal M}_\theta$ for ${\cal C}_\theta$ such that the labels $1$ and $\theta$ are adjacent. Now increase each label of $f'$ by one, then we have two labels $2$ and $\theta+1$ adjacent. By inserting two vertices with labels $1,\theta+2$ between $2$ and $\theta+1$, we have a segment as $\cdots 2 ~\theta+2~1~\theta+1\cdots$, and obtain a labeling $f$ of a cycle ${\cal C}_{\theta+2}$. It is easy to check that
\begin{equation*}
  \begin{aligned}
  {\cal M}(f)&= {\cal M}_\theta+2\sum_{i=1}^{\theta}i+\theta-2(\theta+1)+2(\theta+2)+(\theta+2)+(\theta+1)\\
  &={\cal M}_\theta+\theta^2+4\theta+5.\\
  \end{aligned}
\end{equation*}
Hence, ${\cal M}_{\theta+2}\leq {\cal M}_\theta+\theta^2+4\theta+5$. By Lemma~\ref{cycupp}, we have proved our statement. The exact values can be easily computed from ${\cal M}_{3}$, ${\cal M}_{4}$, and the recursion.
\end{proof}

We further show that the  labeling  in Lemma~\ref{cycexa} for ${\cal C}_\theta$ also maximizes the access-minsum in Problem~\ref{maxminsum}.

\begin{theorem}
There exists a  labeling for ${\cal C}_\theta$ which is optimal for both  Problem~\ref{maxminsum} and Problem~\ref{minvar}.
\end{theorem}

\begin{proof}
Let ${\cal S}_\theta=(V,E)$  be the set system whose line graph is ${\cal C}_\theta$. Note that ${\cal S}_\theta$ is also a cycle of length $\theta$. The vertex-labeling $f$  of ${\cal C}_\theta$  in Lemma~\ref{cycexa} naturally induces an edge labeling $\sigma$ of ${\cal S}_\theta$, which  minimizes the access-variance of ${\cal S}_\theta$ in the MinVar model.

Now we claim that it also maximizes the access-minsum for the MaxMinSum model by induction on $\theta$. It is easy to get that the maximum access-minsum of ${\cal S}_\theta$ is at most $\theta$. When $\theta=3,4$, it is trivially true. Assume that  for all $\theta$ or less the claim is true, let us consider ${\theta +2}$.
By assumption, the edge labeling $\sigma'$ induced by $f'$  in Lemma \ref{cycexa} for cycle ${\cal} C_{\theta}$  maximizes the access-minsum for ${\cal S}_\theta$, i.e., for each edge $v_i \mathop{\sim} v_j$ in ${\cal} C_{\theta}$, we have $f'(v_i)+f'(v_j)\geq \theta$. By the proof of Lemma \ref{cycexa}, the labeling $f$ looks like $\cdots 2 ~\theta+2~1~\theta+1\cdots$, which is obtained from $f'$ by increasing each label by one, then deleting the edge $2\sim \theta+1$ and adding three new edges $2\sim \theta+2$, $\theta+2\sim 1$, $1\sim \theta+1$. It is clear that for all edges  $v_i\sim v_j$ in ${\cal} C_{\theta+2}$, we have $f(v_i)+f(v_j)\geq \theta+2$. Hence, the induced edge labeling $\sigma$  maximizes the access-minsum of  ${\cal S}_{\theta+2}$.
\end{proof}

%

\section{Estimate $\textsf{MinVar}({\cal S})$ in Problem~\ref{minvar}}\label{s:5}

In this section, we estimate the value of $\textsf{MinVar}({\cal S})$   in Problem~\ref{minvar} for some set systems $\cal S$.

\begin{lemma}\label{ext}
  Let $H_1=(V,E_1)$ and $H_2=(V,E_2)$ be two regular graphs on the same vertex set. If $E_1\cap E_2=\emptyset$ and $H_1$ is supermagic,
  then  the minimum access-variance of  $G=(V,E_1\cup E_2)$ satisfies
  $$\textsf{MinVar}(G)\leq \textsf{MinVar}(H_2).$$
\end{lemma}
\begin{proof}
  We prove it by giving an edge labeling $\sigma$ of $G$. For  graph $H_1$, we label the edges by set $[|E_2|+1,|E_1|+|E_2|]$ such that $\sigma^*(v)$ is constant for all $v\in V$. This can be done since $H_1$ is supermagic. Then label edges in $H_2$ by $[1,|E_2|]$ such that the variance is $\textsf{MinVar}(H_2)$. It is easy to compute that $Var(G_\sigma)=\textsf{MinVar}(H_2)$.
\end{proof}

Now consider the access-variance of $K_{4r}$, which is an open case in Question~\ref{q:35}. Further, FR codes based on  complete graphs yield the first class of  MBR codes which have the additional property of exact \emph{repair by transfer} \cite{rashmi2009explicit,shah2011distributed}. To estimate $\textsf{MinVar}(K_{4r})$,
we view $K_{4r}$ as a disjoint union of $rK_4$ and $T(4r,r)$. By Theorem~\ref{tura}, $T(4r,r)$ is supermagic. By Lemma~\ref{ext}, we have $\textsf{MinVar}(K_{4r})\leq \textsf{MinVar}(rK_4)$. Note that the line graph of $rK_4$ is $rT(6,3)$, for which the \textsf{MinPS} has been determined in Lemma~\ref{mtnr}. By the connection of  values ${\cal M}(rT(6,3))$  and $\textsf{MinVar}(rK_4)$ in Lemma~\ref{var}, we can give an upper bound of $\textsf{MinVar}(K_{4r})$.

\begin{lemma}\label{k4r}
For any positive integer $r$, we have \begin{equation*}
\textsf{MinVar}(K_{4r})\leq \begin{dcases}
3r, &r \text{ is odd};\\
7r, &r \text{ is even}.
\end{dcases}
\end{equation*}
\end{lemma}
\begin{proof} By setting parameters in Lemma \ref{var} as $\rho=2,\alpha=3,\theta=6r$, and  parameters in Lemma \ref{mtnr} as $m=r,n=6,r=3$, the value of $\textsf{MinVar}(K_{4r})$ is upper bounded by $32M-72r^2-18r+c$, where $c=-r(6r+1)(30r+7)$ is determined by Lemma \ref{var}, and
\begin{equation*}
M=\begin{dcases}
\frac{45r^3+36r^2+7r}{8}, &r \text{ is odd};\\
\frac{45r^3+36r^2+8r}{8}, &r \text{ is even},
\end{dcases}
\end{equation*}
 is determined by Lemma~\ref{mk} for $rK_3$.
\end{proof}

Combining Lemma~\ref{var} and the averaging upper bound in Lemma~\ref{gebd}, we have $\textsf{MinVar}(K_{4r})=O(r^7)$, which is greatly improved to $O(r)$ by Lemma~\ref{k4r}. It is easy to derive a lower bound to show that $\textsf{MinVar}(K_{4r})=\Theta(r)$. In fact, the average popularity of all nodes in $K_{4r}$ is $(4r-1)(8r^2-2r+1)/2$, whose fractional part is $0.5$. Since the total popularity for each node is an integer, we have $\textsf{MinVar}(K_{4r})\geq 4r\times(0.5)^2=r$. We raise the following question.

\begin{question}
  Whether the upper bound in Lemma~\ref{k4r} is tight? It is true for $r=1$.
\end{question}

Now  we consider  Tur\'{a}n graphs. Let $G=L(T(n,r))$  be the line graph of a Tur\'{a}n graph with $r\mid n$ and $r\geq 2$.  By Question~\ref{q:35}, we only need to deal with the case $r\equiv 0\pmod 4$ and $\frac{n}{r}$ is odd.  Chetwynd and Hilton gave the following property of regular graphs.

\begin{theorem}\label{fact}\cite{1-factorizable}
  Let $G$ be a $d$-regular graph of $2n$ vertices and $d\geq \frac{12}{7}n.$ Then $G$ is $1$-factorable.
  \end{theorem}

By Theorem~\ref{fact}, $T(n,r)$ is $1$-factorable when $r\geq 7$, $n$ is even and $r\mid n$. Note that a perfect matching in $T(n,r)$ will be an independent set of $L(T(n,r))$.  Let $m=\frac{(r-1)n}{r}$, and $d=2\left(m-1\right)$. Then $G=L(T(n,r))=(V,E)$ is an $m$-partite $d$-regular graph, where $V=V_1\cup V_2\cup \cdots \cup V_{m}$ with $|V_i|=\frac{n}{2}, i\in[m]$. Further, for each $i\neq j\in [m],$ the subgraph induced by $V_i\cup V_j$ is a $2$-regular graph.

\begin{question}
 For every $r\geq 7$ and even $n$ satisfying $r\mid n,$  can we determine of \textsf{MinPS} of $G=L(T(n,r))$?
\end{question}

\section{Conclusion}\label{s:6}
Motivated by the DRESS codes and access-balancing problem in distributed storage systems, we propose a new combinatorial model, called MinVar model, which is a problem of labeling blocks of  set systems such that the access-variance is minimized (Problem~\ref{minvar}). This problem can be viewed as a generalization of the magic labeling problem when graphs are not super-magic. We further establish an equivalent problem if the set system is linear, which is a vertex-labeling problem of graphs (Problem~\ref{graphpro}). By solving both problems, we are able to find serval families of optimal FR codes based on special graphs, which have minimum access-variance. Besides their applications in access-balancing issue in distributed storage, we think Problems~\ref{minvar} and \ref{graphpro} are interesting by themselves and worth further study in the future. Especially, we restate Question~\ref{q:35} for more attention.

 \emph{Question 3.6: When $n=r\equiv 0 \pmod 4$, or $r \equiv 0 \pmod 4$ and $\frac{n}{r}$ is odd, what is $\textsf{MinVar}({\cal S})$ for ${\cal S}=T(n,r)$?}

 We have given an upper bound in Lemma~\ref{k4r} for the case $n=r\equiv 0 \pmod 4$, which we think is tight. For the other case, we only  have a weak bound in Lemma~\ref{gebd}.

Next, we briefly discuss the model when the access frequencies of information chunks obey the Zipf law. Using the same notations, Problem \ref{graphpro} can be rewritten as follows.
\begin{problem}\label{graphprozipf}
  Given a $d$-regular graph $G$ with $\theta$ vertices, $v_1,v_2,\cdots,v_{\theta}$ and some $\beta >0$,  find a  bijection $\ell:V(G)\rightarrow \{1,1/2^\beta,\ldots ,1/\theta^\beta\}$,  such that ${\cal \overline{M}}(\ell)=\sum_{v_i\sim v_j}\ell(v_i)\ell(v_j)$ is minimized. Denote the  \textsf{MinPS} of $G$ as ${\cal \overline{M}}(G)=\min_\ell{\cal \overline{M}}(\ell)$. 
\end{problem}

By the similar arguments, we can obtain some parallel results with respect to Problem~\ref{graphprozipf}. First, when $G=K_\theta$, ${\cal \overline{M}}(\ell)$ is a constant for any vertex labeling. For Tur\'{a}n graphs, ${\cal \overline{M}}(T(n,r))$ can be achieved when the $k$-th part of vertices is labeled by $\left\{\frac{1}{((k-1)r+1)^\beta},\frac{1}{((k-1)r+2)^\beta},\ldots,\frac{1}{(kr)^\beta}\right\}$. However, for union of complete graphs or Tur\'{a}n graphs, we need to solve the following set-partition problem to estimate  ${\cal \overline{M}}(mK_r)$ and ${\cal \overline{M}}(mT(n,r))$.

\begin{question}
Find an equipartition of $\{1,1/2^\beta,\ldots,1/(mr)^\beta\}$ into $m$ subsets $S_1,S_2,\ldots,S_m$ in polynomial time, such that $\sum_{i=1}^m\left(\sum_{j\in S_i}j\right)^2$ is minimized.
\end{question}

Finally, we mention that the upper bounds of the maximum file size $A(n,k,\alpha,\rho)$ of $[(\theta, M), k,(n, \alpha, \rho)]$ DRESS in Eqs. (\ref{ub1}) and (\ref{ub2}) are still tight in some cases even if we impose the zero access-variance property. Indeed, Corollary \ref{optMin} provides some examples of zero access-variance optimal FR codes. For general FR codes with a limited access-variance, it is interesting that one can improve the upper bounds in Eqs. (\ref{ub1}) and (\ref{ub2}). However, it is not easy to apply their original proofs in \cite{el2010fractional} to the access-variance case. We leave this problem for future study.

%

\end{document}